\documentclass[aps,pre,twocolumn,amsmath,groupedaddress,amssymb]{revtex4-1}
\usepackage{graphicx}
\usepackage{amsmath}
\usepackage{amsfonts}
\usepackage{amsthm}
\usepackage{subfigure}
\usepackage{hyperref}  

\usepackage{algorithm}
\usepackage{algpseudocode}

\newtheorem{theorem}{Theorem}

\newtheorem{lemma}{Lemma}

\newcommand{\g}{$G$}

\newcommand{\ri}{_{r\to i}}
\newcommand{\sj}{_{s\to j}}

\newcommand{\ind}{\in\partial}
\newcommand{\jnor}{\ind j \setminus r}
\newcommand{\rnoi}{\ind r \setminus i}
\newcommand{\deq}{\stackrel{\mathrm{d}}{=}}

\newcommand{\pha}{\mathsf{phase\ 1}}

\DeclareMathOperator*{\argmax}{arg\,max}

\begin{document}
\title{The statistical mechanics of random set packing and a generalization of the Karp-Sipser algorithm}

\author{C. Lucibello}
\affiliation{Dipartimento di Fisica, Universit\`a ``La Sapienza'', P.le A. Moro 2, I-00185, Rome, Italy}
\author{F. Ricci-Tersenghi}
\affiliation{Dipartimento di Fisica, INFN -- Sezione di Roma1, CNR-IPCF UOS Roma Kerberos, Universit\`a ``La Sapienza'', P.le A. Moro 2, I-00185, Rome, Italy}
\begin{abstract}
We analyse the asymptotic behaviour of random instances of the Maximum Set Packing (MSP) optimization problem, also known as Maximum Matching or Maximum Strong Independent Set on Hypergraphs. We give an analytical prediction of the MSPs size using the 1RSB cavity method from statistical mechanics of disordered systems. We also propose a heuristic algorithm, a generalization of the celebrated Karp-Sipser one, which allows us to rigorously prove that the replica symmetric cavity method prediction is exact for certain problem ensembles and breaks down when a core survives the leaf removal process. The $e$-phenomena threshold discovered by Karp and Sipser, marking the onset of core emergence and of replica symmetry breaking, is elegantly generalized to  $c_s = \frac{e}{d-1}$ for one of the ensembles considered, where $d$ is the size of the sets.
\end{abstract}
\maketitle
\section{Introduction}
The Maximum Set Packing is a very much studied problem in combinatorial optimization, one of Karp's twenty-one  NP-complete problems. Given a set $F=\{1,\ldots,M\}$ and a collection of its subsets $\mathcal{S}=\{S_i \ |  S_i\subseteq F, i\in V\}$ labeled by $V=\{1,\ldots,N\}$, a \textit{set packing} (SP) is a collection of the subsets $S_i$ such that they are pairwise disjoint. The size of a SP $\mathcal{S}'\subseteq \mathcal{S}$ is $|\mathcal{S}'|$. A 
\textit{maximum set packing} (MSP) is a SP of maximum size. The integer programming formulation of the MSP problem reads
\begin{align}
\textit{maximize}\qquad &\sum_{i\in V} n_i 
 \label{defopt}\\ 
\textit{subject to}\qquad &\sum_{i\, :\, r \in S_i} n_i \leq 1 \qquad &\forall r\in F
 \label{deffact}\\
 &n_i = \{0,1\} \qquad &\forall i\in V
 \label{defni}
\end{align}
The MSP problem, also known in literature as the matching problem on hypergraphs or the strong independent set problem on hypergraphs, is a an NP-Hard problem. 
This general formulation however can be specialized to obtain two other famous optimization problems: 
 the restriction of the MSP problem to sets $S_i$ of size 2 corresponds to the problem of maximum matching on ordinary graphs and can be solved in polynomial time \cite{Micali1980}; the restriction where each element of $F$ appears exactly 2 times in $\mathcal{S}$ is the maximum independent set and belongs to the NP-Hard class.
 Another common specialization of the general MSP problem, known as $k$-set packing, is that in which all sets $S_i$ have size at most $k$. This is one of the most studied specializations in the computer science community, 
the efforts concentrating on minimal degree conditions to obtain a perfect matching \cite{Alon1997}, linear relaxations \cite{Furedi1993,Chan2011} and approximability conditions \cite{Rodl1996,Halldorsson2000,Hazan2006}. Motivated by this interest, we choose a $k$-set packing problem ensemble as the principal application of the general analytical framework developed in the following sections. The asymptotic behaviour of random sparse instances of the MSP problem  have not been investigated by mathematicians and computer scientists, only in the matching \cite{Karp1981a} and independent set \cite{Wormald1999} restrictions some work has been done. Extending some theorems of \cite{Karp1981a} (on which a part of this work is greatly inspired) to a greater class of problem ensembles is some of the main aims of the present work.
%NOTE applications missing: si può aggiugere il train routing per il weighted set packing

On the other hand also the statistical physics literature is lacking an accurate study of the random MSP problem. One of its specialization though, the matching problem, has been covered since the beginning of the physicists' interest in optimization problems, with the work of Parisi and M\'ezard on the weighted and fully connected version of the problem \cite{Parisi1985,Parisi1986}. More recently the matching problem on sparse random graphs has also been accurately studied \cite{Zhou2003a,Zdeborova2006} using the cavity method technique. Also  the independent set problem on random graphs \cite{Pin} and the dual problem to set packing, the set covering problem \cite{Tarzia2007}, received some attention by the disordered physics community.
The SP problem was investigated with the cavity method formalism in a disguised form, as a glass model on a generalized bethe lattice, in \cite{Weigt2003,Hansen-Goos2005}. This corresponds, as we shall see in the next section, to a factor graph ensemble with fixed factor and variable degrees, thus we will not cover this case in Section \ref{sec:applications}.

The paper is organized as follows:
\begin{itemize}
\item In Section \ref{sec:statphys} we map the MSP problem (\ref{defopt}-\ref{defni}) into a statistical physical model defined on a factor graph and relate the MSP size to the density $\rho$ at infinite chemical potential.
\item We introduce the replica symmetric (RS) cavity method in Section \ref{sec:rs} and  give an estimate for the average MSPs size on sparse factor graph ensembles in the thermodynamic limit.
\item In Section \ref{sec:rsb} we establish a criterion for the validity of the RS ansatz and introduce the 1RSB formalism.
\item In Section \ref{sec:gks} we propose a generalization of the Karp-Sipser heuristic algorithm \cite{Karp1981a} to the MSP problem and prove the validity of the RS ansatz for certain ensemble of problems. Moreover we find a relationship between a core emergence phenomena and the breaking of replica symmetry breaking.
\item Section \ref{sec:num} describes the numerical simulations performed.
\item In Section \ref{sec:applications} we apply the analytical tools developed to some problem ensembles. We compare the numerical results obtained from an exact algorithm with the analytical predictions, focusing to greater extent to one ensemble modelling the $k$-set packing.
\end{itemize}

\section{Statistical physics description}
\label{sec:statphys}
In order to turn the MSP combinatorial optimization problem into a useful statistical physical model let us recast  eqs. (\ref{defopt}-\ref{defni}) into a graphical model using the factor graph formalism \cite{kschichang01,Montanari2009}. We define our variable nodes set to be $V$ and to each $i\in V$ we associate a variable $n_i$ taking values in $\{0,1\}$ as in eq. (\ref{defni}). $F$ will be our factor nodes set, as its elements acts as hard constrains on the variables $n_i$ through eq. (\ref{deffact}). The edge set $E$ is then naturally defined as $E=\{(i,r) | i\in V, r\in S_i\subseteq F\}$. We call $G=(V,F,E)$ the factor graph thus composed and can then rewrite eq. (\ref{deffact}) as
\begin{equation}
\sum_{i\ind r} n_i \leq 1 \qquad \forall r\in F
\label{constr}
\end{equation}
A SP is a configuration $\{n_i\}$ satisfying eq. (\ref{constr}) and its relative size is $\rho(\{n_i\})=\frac{1}{N}\sum_{i\in V} n_i$ that is simply the fraction of occupied sites.

It is now easy to define an appropriate Gibbs measure for the MSPs problem on \g\ through the grand canonical partition function 
\begin{equation}
\Xi_G(\mu)=\sum_{\{n_i\}}\ \prod_{i\in V} e^{\mu n_i} \prod_{r\in F}\mathbb{I}\Big(\sum_{i \ind r} n_i \leq 1\Big)\ ,
\label{matchmodel}
\end{equation}
Only SPs contribute to the partition function, and in the \textit{close packing limit}, as we shall call the limit $\mu\uparrow+\infty$, the measure is dominated by MSPs.
Eq. (\ref{matchmodel}) is also a model for a particle gas  with hard core repulsion and chemical potential $\mu$ located on a hypergraph, and as such has been studied mainly on lattice structures and more in general on ordinary graphs. 
Model (\ref{matchmodel}) has been studied on a generalized Bethe lattice (that is the ensemble $\mathbb{G}_{RR}(d,c)$ defined in Section \ref{sec:applications}, a $d$-uniform $c$-regular factor graph) in \cite{Weigt2003,Hansen-Goos2005} as a prototype of a system with finite connectivity showing a glassy behaviour. This has been the only approach, although disguised as a hard spheres model, from the statistical physics community to a general MSP problem. 

The grand canonical potential is defined as 
\begin{equation}
\omega_G(\mu)=-\frac{1}{\mu N} \log\Xi_G(\mu)\ ,
\label{grand-can}
\end{equation}
and the particle density as
\begin{equation}
\rho_G(\mu)=\frac{1}{N}<\sum_{i\in V} n_i>_{G,\mu}=-\omega_G(\mu)-\mu\, \partial_\mu \omega_G(\mu)\ .
\end{equation}
Grand potential and density are related to entropy by the thermodynamic relation 
\begin{equation}
s_G(\mu)=-\mu\big(\omega_G(\mu)+\rho_G(\mu)\big)=\mu^2\partial_\mu \omega_G(\mu).
\label{entropy_def}
\end{equation}
In the close packing limit (i.e. $\mu\uparrow + \infty)$ we recover the MSP problem, since in this limit the Gibbs measure is uniformly concentrated on MSPs and $\rho_G$ gives the MSP relative size. Since entropy remains finite in this limit, from eq. (\ref{entropy_def}) we obtain the MSP	 relative size 
\begin{equation}
\rho_G\equiv\lim_{\mu\to+\infty}\rho_G(\mu)=\lim_{\mu\to+\infty}-\omega_G(\mu)
\label{rhoome}
\end{equation}
In this paper we focus on random instances of the MSP problem. As usual in statistical physics we will assume the number of variables $N$ (the number of subsets $S_i$)  and the number of constrains $M$ to diverge, keeping the ratio $\frac{N}{M}$ finite. We shall refer to this limit as the \textit{thermodynamic limit}. Instances of the MSP problem will be encoded in factor graph ensembles which we assume to be locally tree-like in the thermodynamic limit.

The MSP relative size $\rho_G$ is a self-averaging quantity in the thermodynamic limit and we want to compute its  asymptotic value
\begin{equation}
\rho = \lim_{N\to +\infty} \mathbb{E}_G [\,\rho_G\,]\ ,
\label{ro}
\end{equation}
where we denoted with $\mathbb{E}_G [\ \bullet\ ]$ the expectation over the factor graph ensemble.
In last equation the $N$ dependence is encoded in the graph ensemble considered. Computing eq. \eqref{ro} is not an easy task and some approximation have to be taken. We shall employ the cavity method from the statistical physics of disordered systems \cite{Parisi1987,Montanari2009}, using both the replica symmetric (RS) and the one-step replica symmetry breaking (1RSB) ansatz. We will prove in Section \ref{sec:gks} that the RS ansatz is exact in a certain region of the phase space while in Section \ref{sec:applications} we will give numerical evidence that the 1RSB approximation gives very good results outside the RS region.
\section{Replica symmetry}
\label{sec:rs}

\subsection{Bethe approximation on a single instance}
The RS cavity method has been known for many decades outside the statistical physics community as the Belief Propagation (BP) algorithm and only in recent years the two approaches have been bridged \cite{kschichang01,Montanari2009}. We start with a variational approximation to the grand potential eq. \eqref{grand-can} of an instance of the problem, the Bethe free energy approximation:

\begin{equation}
\begin{aligned}
\omega^{RS}_G[\boldsymbol{\hat{\nu}}]=&\frac{1}{N}\left[
\sum_{r\ind F} \omega_r[\boldsymbol{\hat{\nu}}]
			+\sum_{i\in V}(1-|\partial i|)\,\omega_i[\boldsymbol{\hat{\nu}}]
 \right]\ ,
\end{aligned}
\end{equation}
with the factor and variable contributions given by
\begin{equation}
\begin{aligned}
\omega_r[\boldsymbol{\hat{\nu}}]=&-\frac{1}{\mu}\log\left[\sum_{\boldsymbol{n}_{i\in \partial r}}\mathbb{I}\Big(\sum_{i \ind r} n_i \leq 1\Big)\prod_{j\in \partial r}\prod_{s\jnor}\hat{\nu}\sj(n_j)\right] ,\\
\omega_i[\boldsymbol{\hat{\nu}}] =&-\frac{1}{\mu}\log\left[\sum_{n_i}e^{\mu n_i}\prod_{r\in\partial i}\hat{\nu}_{\ri}(n_i)\right]\ .
\end{aligned}
\end{equation}
The grand canonical potential is expressed as a function of the factor node to variable node messages $\boldsymbol{\hat{\nu}}=\{\hat{\nu}\ri\}$. Minimization of $\omega^{RS}_G[\boldsymbol{\hat{\nu}}]$ over the messages constrained to be normalized to one, yields
the fixed point BP equations for the set packing: 
\begin{equation}
\begin{aligned}
\hat{\nu}\ri(1)=&\frac{1}{Z\ri}\prod_{j\rnoi}\prod_{s\jnor}\hat{\nu}\sj(0)\ ,\\
\hat{\nu}\ri(0)=&\frac{1}{Z\ri}\bigg[\prod_{j\rnoi}\prod_{s\jnor}\hat{\nu}\sj(0)\\
&+e^{\mu}\sum_{j\rnoi}
\prod_{s\jnor}\hat{\nu}\sj(1)\\
&\times\prod_{j'\in \partial r\setminus 
 \{i,j\}}\prod_{s'\in \partial j'\setminus r}\hat{\nu}_{s'\to j'}(0)\bigg]\ .
\label{bpT1}
\end{aligned}
\end{equation}
The coeeficients $Z\ri$ are normalization factor.
Equations (\ref{bpT1}) can be simplified introducing the fields $\{t\ri\}$ defined as 
\begin{equation}
\frac{\hat{\nu}\ri(1)}{\hat{\nu}\ri(0)}= e^{-\mu t\ri}\ ,
\label{deft}
\end{equation}
yielding
\begin{equation}
t\ri=\frac{1}{\mu}\log\left[1+\sum_{j\rnoi}e^{\mu(1-\sum_{s\jnor}t\sj)}\right]\ .
\label{bpT2}
\end{equation}
Since we are interested in the close packing limit  to solve the problem we will straightforwardly apply the zero temperature cavity method \cite{Mezard2002}. The related BP equations which can be found as the $\mu\uparrow \infty$ limit of eq. (\ref{bpT2}) read
\begin{equation}
t\ri =\max\left\{0\right\}\cup\left\{1-\sum_{s\jnor}t\sj\right\}_{j\rnoi}\ .
\label{bpT}
\end{equation}
We note that the messages $\{t\ri\}$ are bounded to take values in the interval $[0,1]$ and that if we set the initial value of each $t\ri$ in the discrete set $\{0,1\}$, at each BP iteration all messages will take value either 0 or 1. These values can be directly interpreted as the occupational loss occurring in the subtree $r\to i$ if the subtree is connected to the occupied node $i$. This loss cannot be negative (thus a gain) since we put an additional constrain on the subtree demanding every $j$ neighbour of $r$ to be empty, and cannot be greater than 1 as well, in fact to $t\ri=1$ corresponds the worst case scenario where an otherwise occupied node $j\rnoi$ has to be emptied.      

The Bethe free energy for model (\ref{matchmodel}) on the factor graph \g\ can be expressed 
as a function of the fixed point messages $\{t\ri\}$ as
\begin{equation}
\begin{aligned}
\omega^{RS}_G(\mu)=&-\frac{1}{\mu N}\big[\sum_{r\ind F} \log\big(1+\sum_{j\ind r}e^{\mu(1-\sum_{s\jnor}t\sj)}\big)\\
 				&+\sum_{i\in V}(1-|\partial i|)\log\big(1+e^{\mu(1-\sum_{r\ind i}t\ri)}\big)\big]\ .
\end{aligned}
\label{wfinite}
\end{equation}  
We finally arrive to the  Bethe estimation of the MSPs relative size, taking the close packing limit of  eq. ($\ref{wfinite}$) and using $\omega^{RS}_G=-\rho^{RS}_G$, given by
\begin{equation}
\begin{aligned}
\rho^{RS}_G=\frac{1}{N}\bigg[&\, \sum_{r\in F}\max\{0\}\cup\{1-\sum_{s\jnor}t\sj\}_{j\in r} \\
&+\sum_{i\in V}(1-|\partial i|)\max\{0,1-\sum_{r\in i}t\ri\}\bigg]\ .
\label{wnonrand}
\end{aligned}
\end{equation}
Let us examine the various contribution to eq. (\ref{wnonrand}) since we want to convince ourselves that it exactly counts the MSP size, at least on tree factor graphs.
The term
$(1-|\partial i|)\max\{0,1-\sum_{r\in i}t\ri\}$ contributes with $1-|\partial i|$ to the sum only if all  the incoming $t$ messages are zero. In this case $n_i$ is frozen to 1, i.e. the variable $i$ takes part to all the MSP in $G$. Obviously all the neighbours of a variable frozen to 1 have to be frozen to 0.
To all its $ |\partial i|$ neighbours $r$, the frozen to 1 variable $i$ sends a message  $1-\sum_{s\in \partial i \setminus r}t_{s\to i}=1$, so that we have $|\partial i|$ contributions in the first sum of eq. \eqref{wnonrand} $\max\{0\}\cup\{1-\sum_{s\jnor}t\sj\}_{j\in r}=1$ and the total contribution from $i$ correctly sums up to 1.
If for a certain $i$ we have a total field $\tau_i\equiv 1-\sum_{r\in i}t\ri<0$ (two or more incoming messages are equal to one) the variable is frozen to 0, that is it does not take part to any MSPs. It correctly does not contribute to $\omega^{RS}_G$ since it sends a message $1-\sum_{s\in i \setminus r}t_{s\to i}\leq 0$ to each neighbour $r$,  thus it is not computed in $\max\{0\}\cup\{1-\sum_{s\jnor}t\sj\}_{j\in r}$. 

The third case is the most interesting. It concerns variables $i$ which take part  to a fraction of the MSPs. We shall call them unfrozen variables. The total field on an unfrozen variable $i$ is $\tau_i=0$ (thus we have no contribution from the second sum in eq. (\ref{wnonrand})) and all incoming messages are 0 except for a single $t\ri=1$. To this sole function node $r$ the node $i$ sends a message 1, so that the contribution of $r$ to the first sum is $1$. Actually BP equations impose that $r$ has to have at least another unfrozen neighbour beside $i$. In other terms the function node $r$ says: whatever MSP we consider, one of my neighbours has to be occupied. The corresponding term 
$\max\{0\}\cup\{1-\sum_{s\jnor}t\sj\}_{j\in r}=1$ in eq. \eqref{wnonrand} accounts for that.

The presence of unfrozen variables is the reason why we cannot express the density $\rho_G$ through the formula
\begin{equation}
\rho_G=<\sum_{i\in V}n_i>\ .
\end{equation}
In fact, using the infinite chemical potential formalism,  we cannot compute
\begin{equation}
<n_i>=\lim_{\mu\to +\infty}\frac{e^{\mu(1-\sum_{r\in i}t\ri)}}{1+e^{\mu(1-\sum_{r\in i}t\ri)}}
\label{nini}
\end{equation}
when $\lim_{\mu\to +\infty}1-\sum_{r\in i}t\ri=0$, and we would have to use the $O\Big(\frac{1}{\mu}\Big)$ corrections to the fields $\{t\ri\}$. We bypass the problem using the grand potential $\omega^{RS}_G$ to obtain $\rho^{RS}_G$, also addressing a problem reported in \cite{Gamarnik2008} of extending an analysis suited for weighted matchings and independent sets to the unweighted case.

\subsection{Ensemble averages}
To proceed further in the analysis and since one is often concerned with the average properties of a class of related factor graphs, let us consider the case where the factor graph \g\ is sampled from a locally tree-like factor graph ensemble $\mathbb{G}(N)$.
We shall employ the following notation for the graph ensembles expectations: $\mathbb{E}_{G}[\ \bullet\ ]$ for graphs averages;
$\mathbb{E}_{C_0}[\ \bullet\ ]$ ($\mathbb{E}_{D_0}[\ \bullet\ ]$) for expectations over the factor (variable) degree distribution, which we will sometime call root degree distribution; $\mathbb{E}_{C}[\ \bullet\ ]$ ($\mathbb{E}_{D}[\ \bullet\ ]$) for expectations over the excess degree distribution of factor (variable) nodes conditioned to have at least one adjacent edge, which we will sometime call residual degree distribution. 
The quantities $c$ and $d$ and the random variables $C, C_0, D$ and $D_0$,  are related by 
 
\begin{equation}
\begin{aligned}
\mathbb{P}[C=k] &= \frac{(k+1)\,\mathbb{P}[C_0=k+1]}{c}\ ,\\
\mathbb{P}[D=k] &= \frac{(k+1)\,\mathbb{P}[D_0=k+1]}{d}\ .
\end{aligned}
\end{equation} 
In Section \ref{sec:applications} we  discuss some specific factor graph ensembles, where $C$ and $D$ are fixed to a deterministic value or Poissonian distributed.

With this definitions the distributional equation corresponding to Belief Propagation formula eq. \eqref{bpT} reads
\begin{equation}
T'\deq \max\{0\}\cup\left\{1-\sum_{s=1}^{D_j}T_{sj}\right\}_{j\in \{1,\ldots,C\}}\ ,
\label{deT}
\end{equation}
where $\{D_j\}$ are i.i.d. random residual variable degrees, $C$ is a random residual factor degree and $\{T_{sj}\}$ are i.i.d. random incoming messages. 
Since on a given graph each message $t\ri$ takes value only on $\{0,1\}$ we can take the 
distribution of messages $t$ to be of the form
\begin{equation}
P(t)=p\ \delta(t)+(1-p)\ \delta(t-1)
\label{fpTdef}
\end{equation}
For this to be a fixed point of eq. (\ref{deT}) the parameter $p$ has to satisfy the self-consistent equation 
\begin{equation}
p_*=\mathbb{E}_{C}\left(1-\mathbb{E}_{D}\, p_*^D\right)^C.
\label{fpTp}
\end{equation}

Using (\ref{wnonrand}) and (\ref{fpTp}) we obtain the replica symmetric approximation to the asymptotic MSP relative size
\begin{equation}
\rho^{\small{RS}}=\frac{d}{c}\Big(1-\mathbb{E}_{C_0}\big(1-\mathbb{E}_{D}\,p_*^D\big)^{C_0}\Big)+\mathbb{E}_{D_0} (1-D_0)p^{D_0}_*.
\label{rho}
\end{equation}
It turns out that the RS approximation is exact only when the ratio $\frac{N}{M}$ is sufficiently low. In the SP language eq. (\ref{rho}) holds true only when the number of the subsets among which we choose our MSP is not too big compared to the number of elements of which they are composed.

In the next section we will quantitatively establish the limits of validity of the RS ansatz and introduce the one-step replica symmetry breaking formalism which provides a better approximation to the exact results in the regime of large $\frac{N}{M}$ ratio. In Section \ref{sec:gks} we  prove that eq. \eqref{rho} is exact for certain choices of the factor graph ensembles.

\section{Replica symmetry breaking}
\label{sec:rsb}
\subsection{RS consistency and bugs propagation}
Here we propose two criterions in order to check the consistency of the RS cavity method. If any of those fails 

The first criterion is the assumption of unicity of the fixed point of eq. (\ref{fpTp}) and its dynamical stability under iteration. We restrict ourselves to the subspace 
of distributions with support on $\{0,1\}$, although it is possible to extend the following analysis to the whole space of distributions over $[0,1]$ with an argument based on stochastic dominance  following \cite{Gamarnik2008}. Characterizing the distributions over $\{0,1\}$ as in eq. (\ref{fpTdef}) with a real parameter $p\in[0,1]$, from eq. (\ref{deT}) we obtain the dynamical system
\begin{equation}
p'=\mathbb{E}_{C}\left(1-\mathbb{E}_{D}\, p^D\right)^C\equiv f(p)\ .
\end{equation}
The stability criterion
\begin{equation}
|f'(p_*)|<1
\label{tstab2}
\end{equation}
suggests that the RS approximation to the MSP size eq. (\ref{rho}) is exact as long as eq. (\ref{tstab2}) is satisfied. This statement will be made rigorous in Section \ref{sec:gks}.

The second method we use to check the RS stability, called \textit{bugs proliferation}, is the zero temperature analogous of spin glass susceptibility. We will compute the average number of changing $t$ messages induced by a  change in a single message $t\ri$ ($1\to 0$ or $0\to 1$). This is given by
\begin{equation}
N_{ch}=\mathbb{E}\Big[\sum_{(s,j)}\ \sum_{\substack{
            a_0,a_1\\
            b_0,b_1}}\mathbb{I}\big(t\sj=b_0\to b_1\ |\ t\ri=a_0\to a_1\big)\Big]\ ,
\end{equation}
where $a_0,a_1,b_0,b_1$ take value in $\{0,1\}$. Since a random factor graph is locally a tree, and assuming correlations decay fast enough, last equation can be expressed as
\begin{equation}
N_{ch}=\sum_{s=0}^{+\infty}(\overline{C}\ \overline{D})^s\sum_{\substack{
            a_0,a_1\\
            b_0,b_1}} P(t_s=b_0\to b_1\ |\ t_o=a_0\to a_1\big)\ ,
\end{equation}
where $t_s$ is a message at distance $s$ from the tree root $o$, and we defined the average residual degrees $\overline{C}=\mathbb{E}_C[C]$ and $\overline{D}=\mathbb{E}_D[D]$. The stability condition  $N_{ch}<+\infty$ yields a constraint on the greatest eigenvalue $\lambda_M$ of the transfer matrix $P(b_0\to b_1\ |a_0\to a_1\big)$: 
 
\begin{equation}
\overline{C}\ \overline{D}\ \lambda_M <1.
\label{bug}
\end{equation}

The two methods presented above give equivalent conditions for the RS ansatz to hold true and they simply express the independence for a finite subgraph from the tail boundary conditions. 

\subsection{The 1RSB formalism}
\label{1RSB}We are going to develop the 1RSB formalism for the MSP problem and then apply it in Paragraph \ref{grp} to the ensemble $\mathbb{G}_{RP}$. We will not check the coherence of the 1RSB ansatz through the inter-state and intra-state susceptibilities \cite{Montanari2007}, we are then not guaranteed against the need of further steps of replica symmetry breaking in order to recover the exact solution. Even in the worst case scenario though, when the 1RSB solution is trivially exact only in the RS region and  a fullRSB ansatz is needed otherwise, the 1RSB prediction for 
MSP relative size $\rho$ should be everywhere more accurate than the RS one and possibly very close to the real value.
We shall refer to the textbook of M\'ezard and  Montanari \cite{Montanari2009} for a detailed exposition of the 1RSB cavity method.

Let us fix a factor graph $G$ from a locally tree-like ensemble $\mathbb{G}$. We  call $Q\ri(t\ri)$ the distribution of messages on the directed edge $r\to i$ over the states of the system. We still expect the messages $t\ri$  to take values 0 or 1, so that $Q\ri$ can be parametrized as  
\begin{equation}
Q_{\ri}(t\ri)=q\ri \delta(t\ri)+(1-q\ri)\delta(t\ri-1).
\end{equation}
The 1RSB Parisi parameter $x\in[0,1]$  has to be properly rescaled in order to correctly take the limit $\mu\uparrow\infty$. Therefore we introduce the new 1RSB parameter $y=\mu x$ which stays finite in the close packing limit and takes value in $[0,+\infty)$. 
The reweighting factor $e^{-y \omega^{r\to i}_{iter}}$ is defined as 
\begin{equation}
e^{-y \omega^{r\to i}_{iter}}=\frac{Z_{r\to i}}{\prod_{j\rnoi}\prod_{s\jnor}Z_{s\to j}}.
\end{equation}
Last equation combined with eqs. \eqref{bpT1} and (\ref{deft}) gives  $\omega^{r\to i}_{iter}=-t\ri$. Averaging over the whole ensemble we can then write the zero temperature 1RSB message passing rules (also called Survey Propagation equations):

\begin{equation}
q'\deq \frac{\prod_{j=1}^C(1-\prod_{s=1}^{D_j}q_{sj})}{e^y+(1-e^y) \prod_{j=1}^C(1-\prod_{s=1}^{D_j}q_{sj})},
\label{1rsb1}.
\end{equation}
In preceding equation $\{q_{sj}\}$ are i.i.d.r.v. on $[0,1]$ and, as usual, $C$ is the random variable residual degree and $\{D_j\}$ are random independent factors residual degrees. Fixed points of eq. (\ref{1rsb1}) take the form 
\begin{equation}
P(q)=p_0\delta(q)+p_1\delta(q-1)+p_2 P_2(q),
\end{equation}
where $P_2(q)$ is a continuous distribution on $[0,1]$ and $p_2 = 1 -p_0-p_1$. Parameters $p_0$ and $p_1$ have to satisfy the closed equations
\begin{equation}
\begin{aligned}
p_1=&\mathbb{E}_{C}\big(1-\mathbb{E}_{D}(1-p_0)^D\big)^C, \\
p_0=&1-\mathbb{E}_{C}\big(1-\mathbb{E}_{D}\,p_1^D\big)^C.
\end{aligned}
\label{p1RSB}
\end{equation}
Solutions of (\ref{p1RSB}) with $p_2=0$ correspond  to replica symmetric solutions and their instability marks the onset of a spin glass phase. In this new phase the MSPs are clustered according to the general scenario displayed by constraint satisfaction problems \cite{Krzakala2007}.

From the stable fixed point of equation (\ref{1rsb1}) we can calculate the 1RSB free energy functional $\phi(y)$  as

\begin{equation}
\begin{aligned}
-y\phi(y)=&\frac{d}{c}\ \mathbb{E}\,{\log\Big[(1-e^y)\prod^{C_0}_{j=1}(1-\prod_{s=1}^{D_j}q_{sj})+e^y\Big]}+\\
&+\mathbb{E}\,{(1-D_0)\log\Big[(1-e^y)(1-\prod_{r=1}^{D}q_r)+e^y\Big]},
\label{free1RSB}
\end{aligned}
\end{equation}
and 1RSB density, $\rho_{\text{{\tiny 1RSB}}}(y)=-\frac{\partial y\phi(y)}{\partial y}$, as
\begin{equation}
\begin{aligned}
\rho_{\text{{\tiny 1RSB}}}(y)=&\frac{d}{c}\ \mathbb{E}\,{\frac{e^y\left(1-\prod^{C_0}_{j=1}(1-\prod_{s=1}^{D_j}q_{sj})\right)}{(1-e^y)\prod^{C_0}_{j=1}(1-\prod_{s=1}^{D_j}q_{sj})+e^y}}+\\
&+\mathbb{E}\, {(1-D_0)\frac{e^y\prod_{r=1}^{D_0}q_r}{(1-e^y)(1-\prod_{r=1}^{D_0}q_r)+e^y}},
\label{rho1RSB}
\end{aligned}
\end{equation}
with expectations intended over $\mathbb{G}$ and over fixed point messages $\{q_s\}$ and $\{q_{sj}\}$. Since the free energy functional $\phi(y)$ and the complexity  $\Sigma(\rho)$ are related 
by the Legendre transform 
\begin{equation}
\Sigma(\rho)=-y\rho - y\phi(y),
\label{therm1rsb}
\end{equation}
with $\frac{\partial \Sigma}{\partial \rho}=-y$, 
through (\ref{free1RSB}) we can compute the complexity taking the inverse transform. Equilibrium states, that is MSPs, are selected by
\begin{equation}
\rho_{\text{\tiny{1RSB}}}=\argmax_{\rho} \{\rho :\Sigma(\rho)\geq 0\}.
\end{equation}
or equivalently taking the 1RSB parameter $y$ to be 
\begin{equation}
y_s=\argmax_{y\in [0,+\infty]} \phi(y).
\end{equation}
In the static 1RSB phase we expect $\Sigma(\rho_s)=0$ so that from (\ref{therm1rsb}) we have 
$\phi(y_s)=-\rho_s$. We will see that this is generally true except for the ensemble $\mathbb{G}_{RP}(2,c)$ of Section \ref{sec:applications} , corresponding to maximum matchings on ordinary graphs, where the equilibrium state have maximal complexity and $y_s=+\infty$.  The relation 
     
\begin{equation}
\rho_{\text{{\tiny 1RSB}}}=-\phi(y_s)
\end{equation}
 is always valid though, since for $y_s\uparrow\infty$ complexity stays finite.

\section{A heuristic algorithm and exact results}
\label{sec:gks}
In this section we propose a heuristic greedy algorithm to address the problem of MSP. It is a natural generalization of the algorithm that Karp and Sipser proposed to solve the maximum matching problem on Erd\H{o}s-R\'enyi random graphs \cite{Karp1981a}, therefore we shall call it Generalized Karp-Sipser (GKS).   
Extending their derivation concerning the leaf removal part of the algorithm we are able to prove that the RS prediction for MSP density is exact as long as the stability criterion 
\eqref{tstab2} is satisfied. We will not give the proofs of the following theorems as they are lengthy but effortless extension of those given in \cite{Karp1981a}. 
In order to find the maximum matching on a graph Karp and Sipser noticed that as long as the graph contains a node of degree one (a leaf),  its unique edge has to belong to one of the perfect matchings. 

They considered the simplest randomized algorithm one can imagine: as long as there is any leaf remove it from the graph, otherwise remove a random edge; then iterate until the graph is depleted. They studied the behaviour of this leaf-removal algorithm on random graphs and were able to prove that it grants w.h.p a maximum matching (within an $o(n)$ error). 

To generalize some of their results we need to extend the definition of leaf to that of pendant. 
We call \textit{pendant} a variable node whose factor neighbours all have degree one, except for one at most. Stating the same concept in different words, all  of the neighbours of a pendant  have the pendant itself as their sole neighbour, except for one of them at most. See Figure \ref{img:leafremoval} for a pictorial representation of a pendant (in red).
The GKS algorithm is articulated in two phases: a pendant removal and a random occupation phase.
We give the pseudo-code for the Generalized Karp-Sipser algorithm:
\begin{algorithm}[H]
\caption{Generalized Karp-Sipser (GKS)}
\label{alg:gkarp}
\begin{algorithmic}
\Require a factor graph $G=(V,F,E)$
\Ensure a set packing $V'$
    \State $V'=\emptyset$
    \State add to $V'$ all isolated variable nodes and remove them
    \State from $G$
    \State remove from $G$ any isolated factor node
\While{$V$ is not empty}
	\If{$G$ has any pendant}
		\State choose a pendant $i$ uniformly at random
		\State add $i$ to $V'$
		\State remove 	$i$ from $G$, then remove its factor neighbours 
		\State and their  variable neighbours
	\Else
		\State pick uniformly at random a variable node and  add
		\State  it to $V'$, remove it from $G$, then remove its factor
		\State neighbours and their variable neighbours
	\EndIf
	\State remove from $G$ any isolated factor node
\EndWhile\\
\Return $V'$
\end{algorithmic}
\end{algorithm}
At each step the algorithm prioritize the removal of pendants over that of random variable nodes.
We notice that the removal of a pendant is always an optimal choice in order to achieve a MSP, we have no guarantees though on the effect of the occupation of a random node.
We call $\pha$ the execution of the algorithm up to the point where the first non-pendant variable is added to $V'$. It is trivial to show that $\pha$ is enough to find a MSP on a tree factor graph.  The interesting thing though is that $\pha$ is also able to deplete non-tree factor graphs and find a MSP as long as the factor graphs are sufficiently sparse and large enough.
\begin{figure}
\centering
\includegraphics[width=0.5\textwidth, trim= 50 400 50 200]{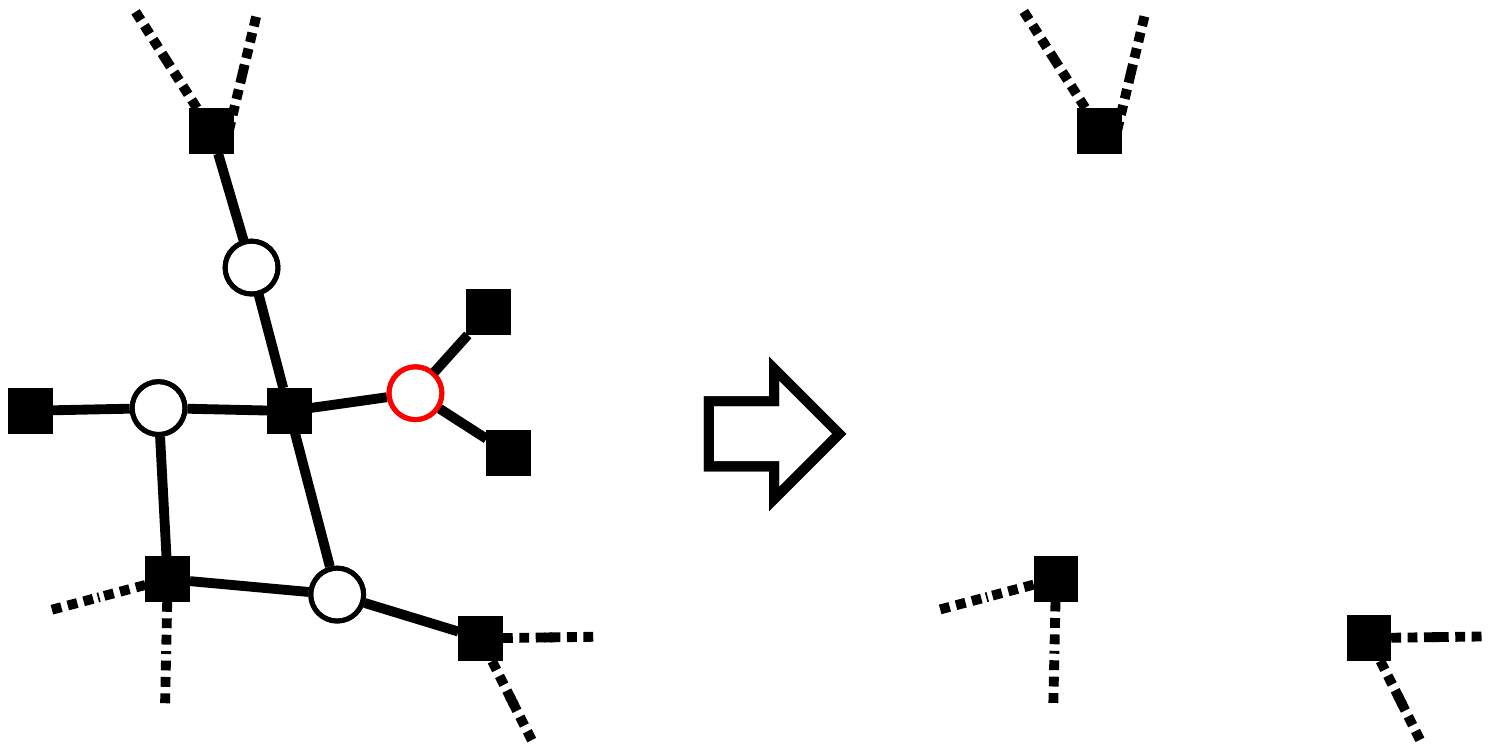}
\caption{The removal of a pendant (depicted in red) as occurs in the inner part of the while loop in GKS algorithm.}
\label{img:leafremoval}
\end{figure}

We call \textit{core} the subset of the variable nodes which has not been assigned to the MSP in   $\pha$. We will show how the emergence of a core is directly related to replica symmetry breaking.  Let us define as usual 
\begin{equation}
f(x)\equiv \mathbb{E}_C\, \left(1-\mathbb{E}_D\, x^D\right)^C.
\label{def-f}
\end{equation}
The function $f$ is continuous, non-increasing and satisfies the relation $1\geq f(0)\geq f(1)= \mathbb{P}\left[C=0\right]$.
It turns out that, in the large graph limit, $\pha$ of GKS  is characterized by the solutions of the system of equations
\begin{equation}
\begin{cases}
p = f(r)\ ,\\
r = f(p)\ .
\end{cases}
\label{rp}
\end{equation}
We notice that system \eqref{rp} is equivalent to 1RSB eqs. \eqref{p1RSB} once the substitution $p\to p_0$ and  $r\to 1 - p_1$ are made.
\begin{lemma}
The system of eqs. \eqref{rp} admits always a (unique) solution $p_*=r_*$.
\end{lemma}
\begin{proof}
By monotony and continuity the function $g(p)=f(p)-p$ has a single zero in the interval $[0,1]$.
\end{proof}
If other solutions are present the relevant one is the one with smallest $p$, as the the following theorems certify.
\begin{theorem}
Let $(p,r)$ be the solution with smallest $p$ of eq. \eqref{rp}. Then the 
density of factor nodes surviving $\pha$ in the large graphs limit is given by
\begin{equation}
\psi_1 = 1+\tilde{p}-\tilde{r}+c\, p\, \mathbb{E}_D\left[p^D-r^D\right]
\end{equation}
with
\begin{equation}
\tilde{r}=\mathbb{E}_{C_0}\left(1-\mathbb{E}_D\, p^D\right)^{C_0}
\end{equation}
and
\begin{equation}
\tilde{p}=\mathbb{E}_{C_0}\left(1-\mathbb{E}_D\, r^D\right)^{C_0}.
\end{equation}
In particular if the smallest solution is $p_*=r_*$ the graph is depleted with high probability in $\pha$.
\label{t1}
\end{theorem}
As already said we will not give the proof of this and of the following theorem, as they are lengthy and can be obtained from the derivation given in Karp and Sipser's article\cite{Karp1981a} with little effort even if not in a completely trivial way.
Theorem \ref{t1} affirms that as soon as eq. \eqref{rp} develops another solution a core survives $\pha$. This phenomena coincides with the need for replica symmetry breaking in the cavity formalism. Let us now establish the exactness the cavity prediction for $\rho$ in the RS phase. 

\begin{theorem}
Let $(p,r)$ be the solution with smallest $p$ of eq. \eqref{rp}. Then the 
density of variables assigned in $\pha$ in the large graphs limit is 
\begin{equation}
\rho_1 = \frac{d}{c}\left(1-\tilde{r}\right)+\mathbb{E}_{D_0} (1-D_0) \, p^{D_0},
\end{equation}
where $\tilde{p}$ has been defined in previous theorem. In particular if the smallest solution is $p_*=r_*$ the replica symmetric cavity method prediction eq. \eqref{rho} is exact.
\label{t2}
\end{theorem}
These two theorems imply that the RS prediction for $\rho$, eq. \eqref{rho}, holds true as long as $\pha$ manages to deplete w.h.p. the whole factor graph. A similar behaviour has been observed in other combinatorial optimization problem, e.g. the random XORSAT \cite{Mezard2003}.
Conversely  it is easy to prove, given the equivalence between eqs. \eqref{rp} and \eqref{p1RSB}, that if  a solution with $p\neq p*$ exists the RS fixed point is unstable in the 1RSB distributional space. Therefore the system is in the RS phase if and only if eqs. \eqref{p1RSB} admit a unique solution.

We notice that we could have also looked for the 
solutions of the single equation $p = f^2(p)$ instead of 
the two of eq. \eqref{rp}, since the value of $r$ is uniquely 
determined by the value of $p$. Then previous theorems 
states that the RS results holds as long as $f^2$ has a 
single fixed point. In  
\cite{Gamarnik2008} it has been proven that the RS solution of the weighted 
maximum independent set and maximum weighted matching  
holds true as long as the corresponding squared cavity operator $f^2$ has a unique fixed point. Since those are special cases of the weighted MSP problem, we conjecture that the $f^2$ condition holds for the general case too, as it does for the unweighted case. Probably it would be possible to further generalize the Karp-Sipser algorithm to cover analytically the weighted case as well.

The scenario arising from the analysis of the algorithm and from 1RSB cavity method is the following: in the RS phase maximum set packings form a single cluster and it is possible to connect any two of them with a path involving MSPs separated by a rearrangement of a finite number of variables \cite{Semerjian2008}. In the RSB phase instead MSPs can be grouped into many connected (in the sense we mentioned before) clusters, each one of them defined by the assignation of the variables in the core. Two MSPs who differ on the core are always separated by a global rearrangement of the variables (i. e. $O(N)$). In presence of a core the GKS algorithm makes some suboptimal random choices (after $\pha$).

We did not make an analytical study of the random variable removal part of the GKS algorithm. This has been done by Karp and Sipser for the matching problem, associating to the graph process a set of differential equations amenable to analysis \cite{Wormald1999} (something similar has been done for random XORSAT as well \cite{Mezard2003}). In that case turns out that the algorithm still achieves optimality and yields an almost perfect matching on the core. An appropriate analysis of the second phase of the GKS algorithm, and the reason of its failure for most of the SP ensembles we covered, deserves further studies.

\section{Numerical Simulations}
\label{sec:num}

While the RS cavity eqs. \eqref{fpTp} and \eqref{rho} can be easily computed, the numerical solution of the 1RSB density evolution eq. \eqref{1rsb1} is much more involved (although simplified by the zero temperature limit) and has been obtained with a standard population dynamics algorithm \cite{Me2001}. 

The implementation of the GKS algorithm is straightforward. While it is very fast during $\pha$, we noticed a huge slowing down in the random removal part. We were able to find set packings for factor graphs with hundred of thousands of nodes.

In order to test the cavity and GKS predictions, we also computed the exact MSP size on factor graphs of small size. First we notice that a set packing problem, coded in a factor graph structure $F$, is equivalent to an independent set problem on an appropriate ordinary graph $G$. The node set of $G$ will be the variable node set of $F$ and we add an edge to $G$ between each pair of nodes having a common factor node neighbour in $F$. Therefore each neighborhood of a factor node in $F$ forms a clique in $G$. We then solve the maximum set problem for $G$ using an exact algorithm \cite{Tsukiyama1977} implemented in the \textit{igraph} library \cite{Csardi}. Since the time complexity is exponential in the size of the graph we performed our simulations on graphs containing up to only one hundred nodes.

\section{Applications to problem ensembles}
\label{sec:applications}
We shall now apply the methods developed in the previous sections to some factor graph ensembles, each modelling a class of MSP problem instances. We consider graph ensembles containing nodes with Poissonian random degree or regular degree: $\mathbb{G}_{PR}(N,d,c)$, $\mathbb{G}_{RP}(N,d,c)$, $\mathbb{G}_{PP}(N,d,c),$ $\mathbb{G}_{RR}(N,d,c)$. Subscript $R$ or $P$ indicates whether the type of nodes to which they refer (variable nodes for the first subscript, factor nodes for the second) have regular or Poissonian random degree respectively. We parametrize these ensemble by their average variable and factor degree, that is $d = \mathbb{E}_{D_0}D_0$ and $c =\mathbb{E}_{C_0} C_0$. They are constituted by factor graphs having $N$ variable nodes and $M=\lfloor\frac{N d}{c}\rfloor$ factor nodes but they differ both for their elements and for their probability law. We define the ensembles giving the probability of sampling one of their elements:

\begin{itemize}
\item $\mathbb{G}_{RP}(N,d,c)$: Each element $G$ has $N$ variables and $M=\lfloor\frac{N d}{c}\rfloor$ factors. Every variable node has fixed degree $d$. $G$ is obtained linking each variable with $d$ factors chosen uniformly and independently at random. The factor nodes degree distribution obtained is Poissonian of mean $c$ with high probability. This ensemble is a model for the $k$-set packing  and will be the main focus of our attention.
\item $\mathbb{G}_{PR}(N,d,c)$: Each element $G$ has $N$ variables and $M=\lfloor\frac{N d}{c}\rfloor$ factor. Every factor node has fixed degree $c$. $G$ is built linking each factors with $c$ variables chosen uniformly and independently at random. The variable nodes degree distribution obtained is Poissonian of mean $d$ with high probability.  
\item $\mathbb{G}_{PP}(N,d,c)$: Each element $G$ has $N$ variables and $M=\lfloor\frac{N d}{c}\rfloor$ factors. $G$ is built adding an edge $(i,r)$ with probability $\frac{c}{N}$ independently for each choice of a variable $i$ and a factor $r$. The factor graph obtained has w.h.p Poissonian variable nodes degree distribution of mean $d$ and Poissonian factor nodes degree distribution of mean $c$.
\item $\mathbb{G}_{RR}(N,d,c)$: It is constituted of all factor graphs of $N$ variable nodes of degree $d$ and $M=\frac{N d}{c}$ factor nodes of degree $c$ ($N d$ has to be multiple of $c$). Every factor graph of the ensemble is equiprobable and can be sampled using a generalization of the configuration model for random regular graph \cite{Wormald1999a}. This ensemble is a model for the $k$-set packing.
\end{itemize}
We shall omit the argument $N$ when we refer to an ensemble in the $N\uparrow\infty$ limit.
\subsection{$\mathbb{G}_{RP}(d,c)$} 
\label{grp}
This is the ensemble with variable node degrees fixed to $d$ and Poissonian factor node degrees, that is $C\sim C_0 \sim \mathrm{Poisson}(c)$. The case $d=2$ corresponds to the maximum matching problem on Erd\"os-R\'enyi random graph \cite{Karp1981a,Zhou2003a}.

Density evolution eq. (\ref{deT}) for $\mathbb{G}_{RP}$ reduces to 
\begin{equation}
 T'\deq \max\{0\}\cup\left\{1-\sum_{s=1}^{d-1}T_{sj}\right\}_{j\in \{1,\ldots,C\}}\ .
 \label{detmess}
\end{equation}
Considering distributions of the form $P(t)=p\delta(t)+(1-p)\delta(1-t)$ fixed points of eq. (\ref{detmess}) reads 
\begin{equation}
p_*=e^{-cp_*^{d-1}}\equiv f(p_*).
\label{HSfp}
\end{equation}
The last equation admits only one solution for each value of $d$ and $c$, as it is easily seen through a monotony argument considering the left and right hand side of the equation.
The values of $c$ as a function of $d$ satisfying $|f'(p_*)|=1$ are the critical points $c_s(d)$ delimiting the RS phase ($c<c_s(d))$, and are given by
\begin{equation}
c_s(d) = \frac{e}{d-1}\ .
\label{rpcrit}
\end{equation}
For  $c>\frac{e}{d-1}$ a core survives $\pha$ of the GKS algorithm as stated by Theorem \ref{t1} and showed in Figure \ref{img:ks_d3}. In the matching case, i.e. $d=2$ eq. \eqref{rpcrit} expresses the notorious $e$-phenomena discovered by Karp and Sipser, while for higher values of $d$ provides an extension of the critical threshold.

We can recover the same critical condition eq. (\ref{rpcrit}) through the bug propagation method, as the transfer matrix  $P(b_0\to b_1\ |\ a_0\to a_1)$ has non zero elements only the off-diagonal:
\begin{equation}
\begin{aligned}
P(0\to 1\ |\ 1\to 0)=&p^{d-1}\ ,\\
P(1\to 0\ |\ 0\to 1)=&p^{d-1}\ ,\\
\end{aligned}
\end{equation}
which give  $\lambda_M=p^{d-1}$. The average branching factor is $\overline{C}\,\overline{D} = c(d-1)$ so that eqs. (\ref{bug}) and \eqref{HSfp} yield eq. \eqref{rpcrit}. 
The analytical value for the relative size of MSPs, that is the particle density $\rho$, is
\begin{equation}
\rho(d,c)=\frac{d}{c}(1-p_*) +(1-d)\, p_*^d \qquad\text{for}\quad  c< c_s(d)\ .
\label{rhoT}
\end{equation}  
In Figure \ref{img:ks-bp} we compared  $\rho$ from eq. \eqref{rhoT} as a function of $c$ for some values $d$ with an exact algorithm applied to finite factor graphs (as explained in Section \ref{sec:num}),  both above and below $c_s$. Clearly for $c>c_s$ the RS approximation is increasingly more inaccurate.

\begin{figure}
\centering
\includegraphics[width=0.5\textwidth]{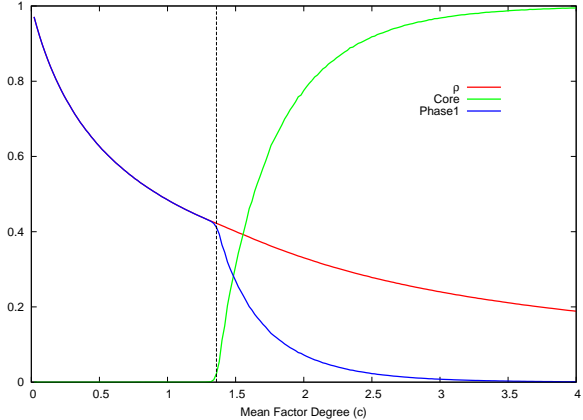}
\caption{Results from GKS algorithm applied to $\mathbb{G}_{RP}(3,c)$. The matching size after $\pha$ and after the algorithms stops are reported. In red is the density of variable nodes left after $\pha$, that is those which belongs the \textit{core}. The vertical line is $c_s=\frac{e}{2}$ and is drawn as a visual aid to determine the point where RS assumptions stop to hold true. Above $c_s$ a non-empty core emerges continuously.}
\label{img:ks_d3}
\end{figure}

\begin{figure}
\centering
\includegraphics[width=0.5\textwidth]{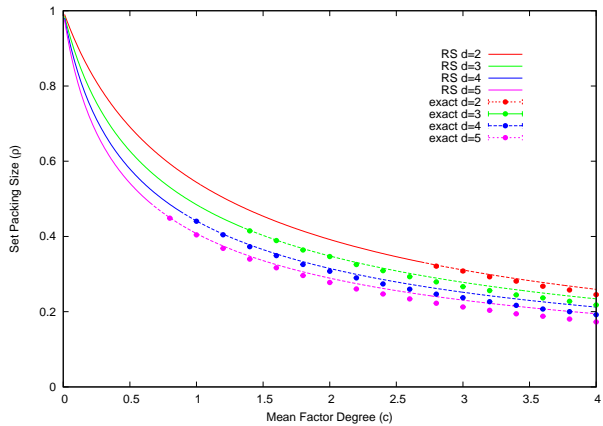}
\caption{RS cavity method analytical prediction eq. \eqref{rhoT} for MSP size in $\mathbb{G}_{RP}(d,c)$ compared with the average size of MSPs obtained from 1000 samples of factor graphs with 100 variable nodes (\textit{dots}). Dashed parts of the lines are the RS estimations in the RSB phase, i.e. for  $c>\frac{e}{d-1}$, where it is no longer exact.}
\label{img:ks-bp}
\end{figure}

We continue our analysis of $\mathbb{G}_{RP}$ above the critical value $c_s$ through the 1RSB cavity method as outlined in paragraph \ref{1RSB}. Fixed point messages 
of (\ref{1rsb1}) are distributed as
\begin{equation}
P(q)=p_0\delta(q)+p_1\delta(q-1)+p_2 P_2(q)
\end{equation}
where 
\begin{equation}
\begin{aligned}
p_1=&e^{-c(1-p_0)^{d-1}}\, \\
p_0=&1-e^{-c p_1^d}\ ,\\
p_2=&1-p_0-p_1\ ,
\end{aligned}
\label{1rsbRP}
\end{equation}
and $P_2(q)$ has to be determined through eq. \eqref{1rsb1}. Equation \eqref{1rsbRP} admits always an RS solution  $p_1=1-p_0=p_*$ which is stable up to $c_s$, as already noticed. Above $c_s$  a new stable fixed point, with $p_2>0$, continuously arises and we study it numerically with a population dynamics algorithm.

The 1RSB free energy, as a function of the Parisi parameter $y$, takes the form 
\begin{equation}
\begin{aligned}
\phi(y)=&-\frac{d}{y c}\ \mathbb{E}\,{\log\Big[(1-e^y)\prod^{C}_{j=1}(1-\prod_{s=1}^{d-1}q_{sj})+e^y\Big]}+\\
&+\frac{d-1}{y}\ \mathbb{E}\,{\log\Big[(1-e^y)(1-\prod_{r=1}^{d}q_r)+e^y\Big].}
\end{aligned}
\label{1RSBfree}
\end{equation}
As prescribed by the cavity method, the value $y_s$ which maximizes $\phi(y)$ over $[0,+\infty]$ yields the correct free energy, therefore we have $\phi(y_s)=-\rho_{\text{{\tiny 1RSB}}}$.

Unsurprisingly, as they belong to different computational classes, the cases $d=2$ and $d\geq 3$ show qualitatively different pictures.
In the case of maximum matching on the Poissonian graph ensemble, numerical estimates suggest that complexity is an increasing function of $y$ on the whole real positive axis. Correct choice for parameter $y$  is then $y_s=+\infty$, as already conjectured in \cite{Zhou2003a}, and we find that maximum matching size prediction from 1RSB cavity method fully agrees with rigorous results from Karp and Sipser \cite{Karp1981a} and with the  size of the matchings given by their algorithm (see Figure \ref{img:1rsb-d2}). The 1RSB ansatz is therefore exact for $d=2$.

\begin{figure}
\centering
\includegraphics[width=0.5\textwidth]{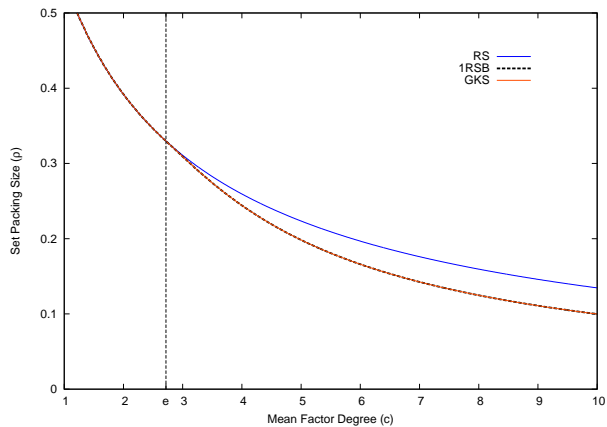}
\caption{Maximum matching size as a function of $c$ for $d=2$. Continuous line corresponds to RS analytical value eq. (\ref{rhoT}), data points are obtained from Karp-Sipser Algorithm and from 1RSB solution eq. (\ref{1RSBfree}) with $y\gg 1$. Karp-Sipser and 1RSB cavity method yield the same and exact result.}
\label{img:1rsb-d2}
\end{figure}

The $d\ge 3$ case  analysis does not yield such a definite result. The complexity of states $\Sigma$ is no more a strictly increasing function of $y$. It reaches its maximum in $y_d$, the choice of $y$ that selects the most numerous states, which could be those where local search greedy algorithms are more likely to be trapped. Then it decreases up to the finite value $y_s$ where complexity changes sign and takes negative values. Therefore $y_s$ is the correct choice for the Parisi parameter which maximizes $\phi(y)$. Plotted as a function of $\rho$, complexity $\Sigma$ has a convex non-physical part, with extrema the RS solution (on the right) and the point corresponding to the dynamic 1RSB solution (on the left), and a concave physically relevant for $y\in[y_d,y_s]$ (see Figure \ref{img:compl-d3}). The 1RSB seems to be in very good agreement with the exact algorithm and we are inclined to believe that no further steps of replica symmetry breaking are needed in this ensemble. The GKS algorithm instead falls short of the exact value, therefore it constitutes a lower bound which is not strict but at least it could probably be made rigorous carrying on the analysis of the GKS algorithm beyond $\pha$.

\begin{figure}
\centering
\subfigure{\includegraphics[width=0.5\textwidth]{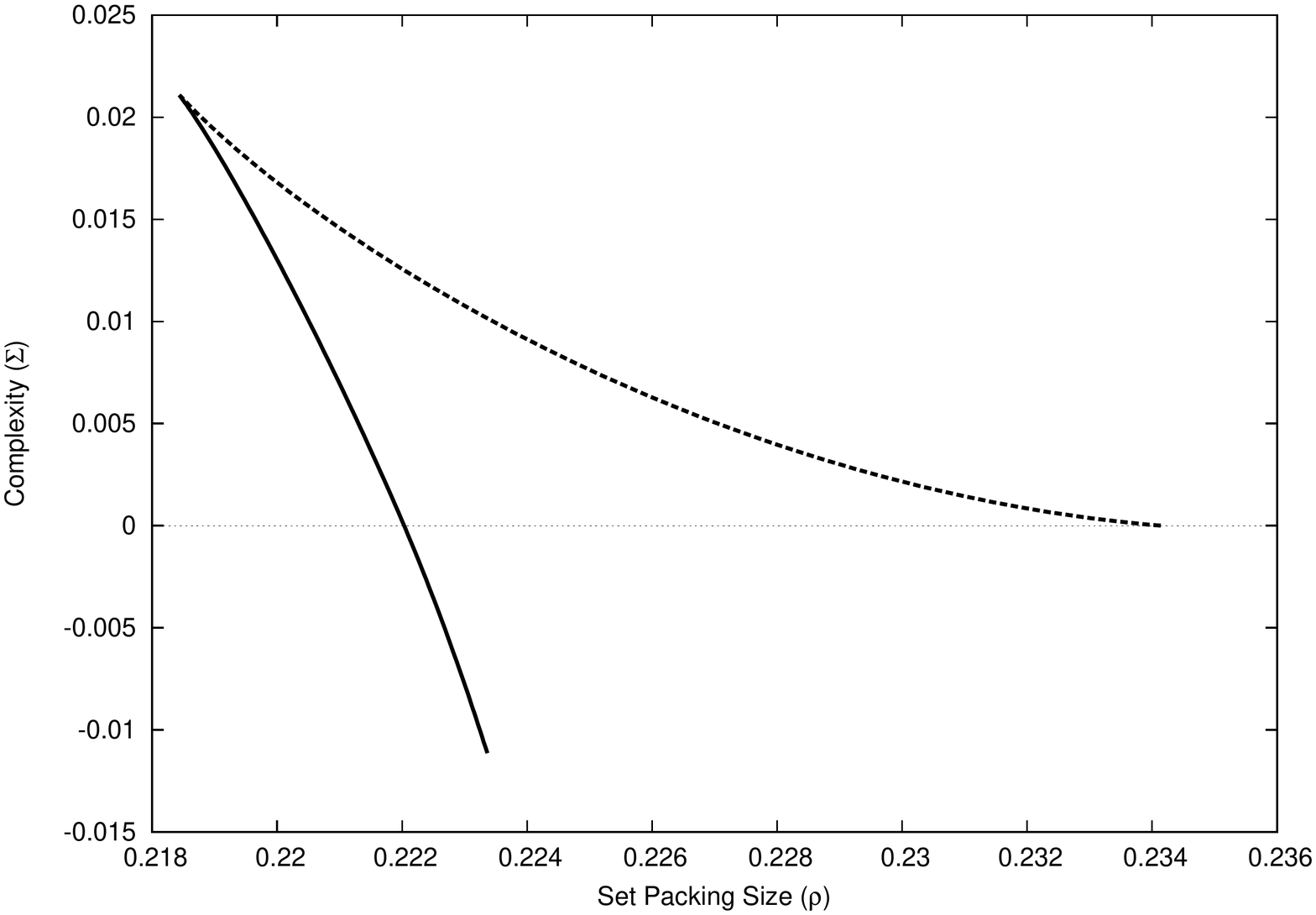}}
\caption{Complexity in $\mathbb{G}_{RP}(d,c)$, with $d=3$ and $c=4.0$, as a function of the relative MSP size $\rho$.}
\label{img:compl-d3}
\qquad
\subfigure{\includegraphics[width=0.5\textwidth]{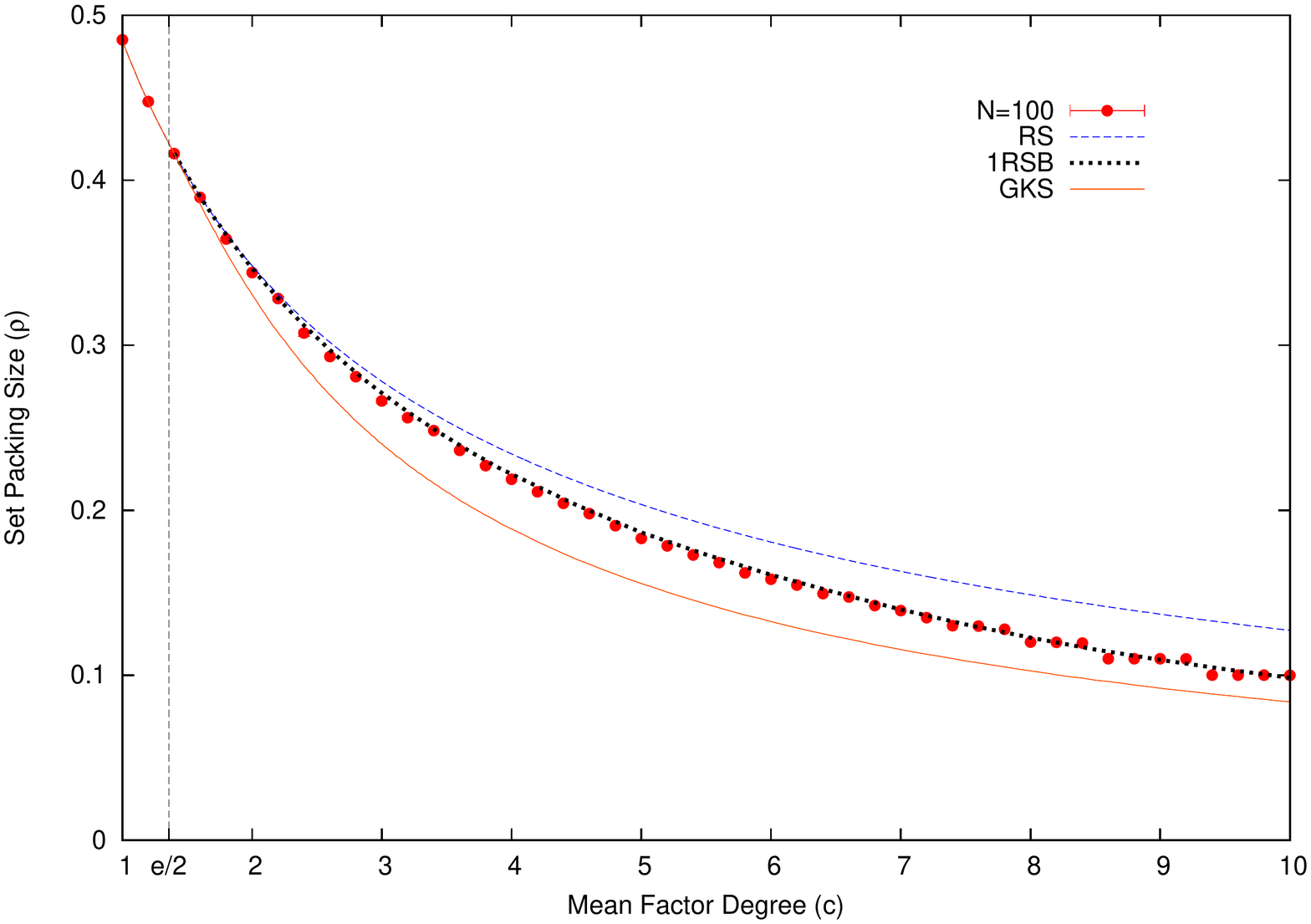}\label{img:1rsb-d3}}
\caption{MSP size in $\mathbb{G}_{RP}(3,c)$ as a function of $c$. RS and 1RSB cavity method are confronted with the GKS algorithm and with an exact algorithm on factor graphs of 100 variable nodes.}
\end{figure}

\subsection{$\mathbb{G}_{PR}(d,c)$}
The ensemble $\mathbb{G}_{PR}(d,c)$ is constituted of factor graphs containing factor nodes of degree fixed to $c$ and variable nodes of Poissonian random degree of mean $d$. It has statistical properties with respect to the MSP problem quite different from those encountered in
$\mathbb{G}_{RP}$, as we will readily show. The MSP problem on $\mathbb{G}_{PR}(d,c)$ with $d=2$ is equivalent to the well known problem of maximum independent sets on Poissonian graphs \cite{Pin,Sanghavi2009}. 
The real parameter $p$ characterizing discrete support distributions of messages has to satisfy the fixed point eq. (\ref{fpTp}), that reads
\begin{equation}
p_*=\big(1-e^{d(p_*-1)}\big)^{c-1}\ .
\label{PRfp}
\end{equation} 
At fixed value of $c$ the RS ansatz holds up to the critical value $d_s(c)$, which is implicitly given by first derivative condition ($\ref{tstab2}$):
\begin{equation}
(c-1)p_*^{\frac{c-2}{c-1}}e^{(p_*-1)d_s} d_s =1\ .
\label{PRcrit}
\end{equation}

Although critical condition eq. \eqref{PRcrit} is not as elegant as the one we obtained for the ensemble $\mathbb{G}_{RP}$, it can be easily solved numerically for $d_s$ as a function of $c$. For $c=2$ the threshold value is exactly $d_s(2)=e$. For $c>2$ instead $d_s$ is an increasing function of the factor degree $c$.
Thanks to eq. (\ref{rho}) we readily compute the MSP  size in the RS phase:

\begin{equation}
\rho=\frac{d}{c}\Big(1-p_*^{\frac{c}{c-1}}\Big)+(1-p_* d)e^{d(p_*-1)} \quad\text{for}\quad  c< c_s(d)\ .
\label{rhoPR}
\end{equation}
We can see in Figure \ref{img:PRrho} that the MSP size $\rho(d,c)$ is a decreasing function in both arguments as expected. Equation \eqref{rhoPR} can be taken as the RS estimate for MSP size for $d>d_s(c)$. The RS estimate is strictly greater than the average size of SPs given by the GKS algorithm at all values of $d>d_s(c)$ (see Figure 7). 
\begin{figure}
\centering
\includegraphics[width=0.5\textwidth]{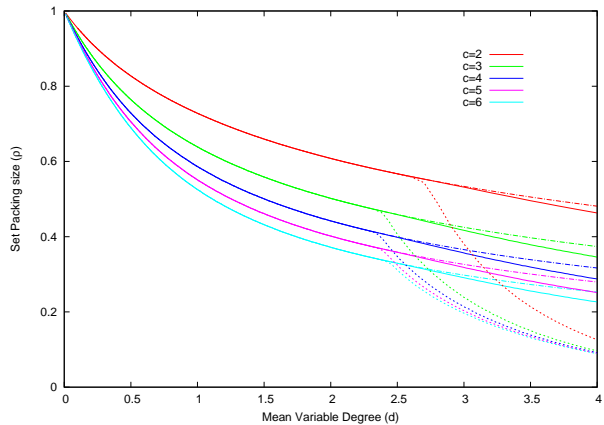}
\caption{RS cavity method analytical prediction eq. \eqref{rhoPR} for MSP size in $\mathbb{G}_{PR}(d,c)$ (\textit{dot-dashed}) compared with the set packing size from $\pha$ of the GKS algorithm (\textit{dotted}) and with a complete run of the algorithm (\textit{continuous}).}
\label{img:PRrho}
\end{figure}

\subsection{$\mathbb{G}_{PP}(d,c)$}
\begin{figure}[h]
\centering
\includegraphics[width=0.5\textwidth]{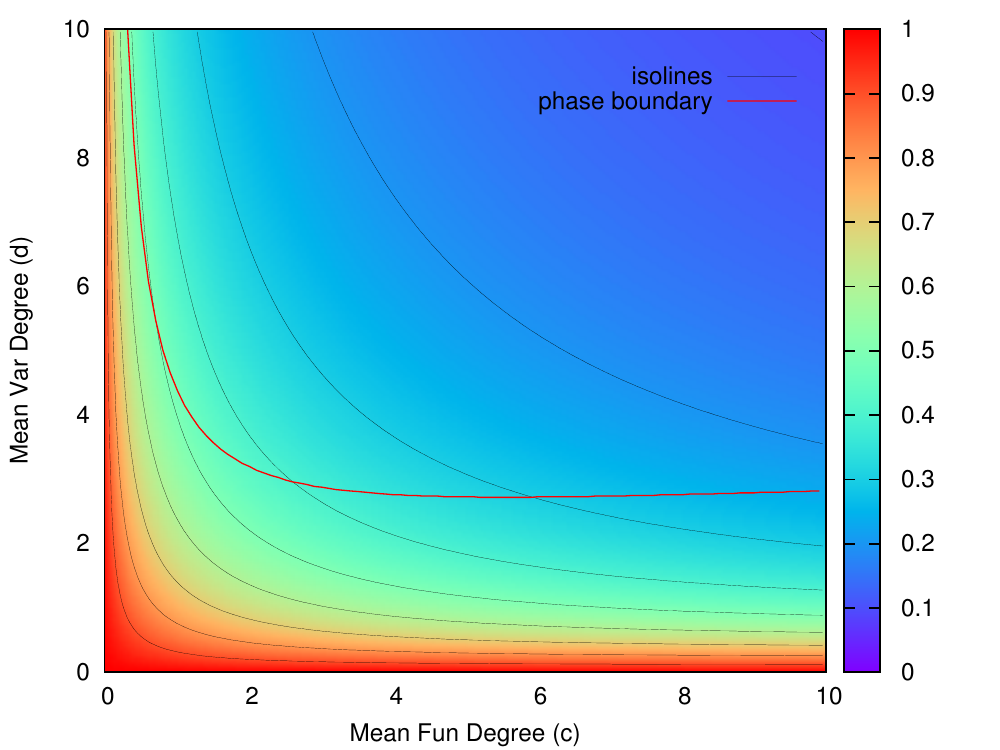}
\caption{MSP relative size for the ensemble  $\mathbb{G}_{PP}(d,c)$ as given by eq. (\ref{rhoPP}). Above the boundary separating RS and RSB phase values are only approximation to the exact value of $\rho$.}
\label{img:PPrho}
\end{figure}
We shall now briefly examine our MSP model on the ensemble $\mathbb{G}_{PP}(d,c)$ where both factor nodes and variable nodes have Poissonian random degrees of mean $c$ and $d$ respectively. From eqs. (\ref{fpTdef}) and (\ref{fpTp}) we obtain the fixed point condition for the probability distribution of messages:
\begin{equation}
p_*=e^{-c e^{d(p_*-1)}}\ .
\label{pPP}
\end{equation}
As usual $p$ is the parameter characterizing the distribution of messages,$P(t)=p\delta(t)+(1-p)\delta(1-p)$. Equation (\ref{pPP}) admits one and only one fixed point solution $p_*$ for each value of $d$ and $c$. In fact $f$ is continuous, strictly decreasing and $f(0)>0,\ f(1)<1$. 
The first derivative condition $|f'(p)|=1$ defines the critical line $d_s(c)$ through
\begin{equation}
  d_s(c) =-\frac{1}{p_* \log( p_*)}\ .
\end{equation}
The curve $d_s(c)$ separates the RS phase from the RSB phase in the $c-d$ parametric space (see Figure \ref{img:PPrho}). The unbounded RS region  shares some resemblance with the corresponding (although $d$-discretized) region of $\mathbb{G}_{PR}(d,c)$ and is at variance with the compact area of the RS phase in $\mathbb{G}_{RP}$.

We can compute the MSP relative size $\rho$ through eq. (\ref{rho}) and obtain
\begin{equation}
\rho=\frac{d}{c}(1-p_*)+(1-d p_*)e^{d(p-1)}
\quad\text{for}\quad  d< d_s(c)\ ,
\label{rhoPP}
\end{equation}
which holds only as an approximation in the RSB phase.
We can see from Figure \ref{img:PPrho} that $\rho$ is decreasing both in $c$ and $d$ as was observed in the other ensembles as well, and

\subsection{$\mathbb{G}_{RR}(d,c)$}
\label{par_grr}
The MSP problem on the ensemble $\mathbb{G}_{RR}(d,c)$, the straightforward generalization to factor graphs of the random regular graphs ensemble, poses some simplification to the cavity formalism thanks to his homogeneity. It  has already been  object of a preliminary studies  by M. Weigt and A. Hartmann \cite{Weigt2003} and then a much more deep work of M. Weigt and H. Hansen Goos \cite{Hansen-Goos2005} who disguised it as a hard spheres model on a generalized Bethe lattice. The authors studied through the cavity method this hard spheres model on $\mathbb{G}_{RR}(d,c)$ both at finite chemical potential $\mu$ and in the close packing limit and found out that the 1RSB solution  is unstable in the close packing limit, therefore suggesting the need of a fullRSB treatment of the problem. 

\section{Conclusions}
We studied the average asymptotic behaviour of random instances of the maximum set packing problem, both from a mathematical and a physical viewpoint. We contributed to the known list of models where the replica symmetric cavity method can be proven to give exact results, thanks to the generalization of an algorithm (and of its analysis) first proposed by Karp and Sipser \cite{Karp1981a}. Moreover, our analysis address a problem reported in recent work on weighted maximum matchings and independent sets on random graphs\cite{Gamarnik2008}, where the authors could not extend their results to the unweighted cases. We achieve here the desired result making use the grand canonical potential instead of the direct computation of single variable expectations. We also extend their condition for the system to be in what physicists call a replica symmetric phase, namely the uniqueness of the fixed point of the square of a certain operator (which is the analogue of the one defined in eq. \eqref{def-f}), to the more general setting of maximum set packing (although without weights).
On some problem ensembles, where the assumptions of Theorems \ref{t1} and \ref{t2} no longer hold and the RS cavity method fails, we used the 1RSB cavity method machinery to obtain an analytical estimation of the MSP size. Numerical simulations show very good agreement of the 1RSB estimation with the exact values, although comparisons have been done only with small random problems due to the exact algorithm being of exponential time complexity. The GKS algorithm instead fails in general to find MSPs but in some special cases.

Some questions remain open to further investigation. To validate the 1RSB approach the stability of the 1RSB solution has to be checked against more steps of replica symmetry breaking. Moreover a thorough analysis of the second phase of the GKS algorithm could shed some light on the mechanism of replica symmetry breaking and give a rigorous lower bound to  the average maximum packing size.

\section*{Acknowledgements}
This research has received financial support from the European
Research Council (ERC) through grant agreement No. 247328 and from the
Italian Research Minister through the FIRB project No. RBFR086NN1.

\bibliography{bibliography}

\end{document}